\documentclass[a4paper,english]{article}

\usepackage{amsmath, amsthm, amssymb}
\usepackage{enumitem}
\usepackage{multirow}
\usepackage[margin=1in]{geometry}
    
\newtheorem{theorem}{Theorem}
\newtheorem{lemma}[theorem]{Lemma}
\newtheorem{corollary}[theorem]{Corollary}
\newtheorem{proposition}[theorem]{Proposition}
\newtheorem*{theorem*}{Theorem}

\makeatletter
\newtheorem*{rep@theorem}{\rep@title}
\newcommand{\newreptheorem}[2]{%
\newenvironment{rep#1}[1]{%
 \def\rep@title{#2 ##1}%
 \begin{rep@theorem}}%
 {\end{rep@theorem}}}
\makeatother
\newreptheorem{theorem}{Theorem}

\theoremstyle{definition}

\theoremstyle{remark}

\newcommand{\eps}{\varepsilon}

\newcommand{\R}{\ensuremath\mathbb{R}}

\newcommand{\G}{\ensuremath\mathcal{G}}

\newcommand{\D}{\ensuremath\mathcal{D}}
\newcommand{\LL}{\ensuremath\mathcal{L}}
\newcommand{\bS}{\ensuremath\mathbb{S}}
\newcommand{\eqdist}{\overset{\mathrm{dist}}{=\joinrel=}}

\DeclareMathOperator{\rk}{rank}

\DeclareMathOperator{\tr}{tr}
\DeclareMathOperator{\diag}{diag}
\DeclareMathOperator*{\E}{\mathbb{E}}
\DeclareMathOperator{\Unif}{Unif}

\allowdisplaybreaks

\title{Tight Bounds for Sketching the Operator Norm, Schatten Norms, and Subspace Embeddings\footnote{This paper appeared in the Proceedings of RANDOM/APPROX 2016, LIPIcs Vol.\ 60, 39:1--39:11. The current version corrects the proof of Corollary 7.

 Yi Li was supported by ONR grant N00014-14-1-0632 when he was at Harvard University, where the major part of this work was done. David P.\ Woodruff was at IBM Research Almaden when this work was done.}}

\author{	Yi Li\\
		School of Physical and Mathematical Sciences \\
		Nanyang Technological University\\ \texttt{yili@ntu.edu.sg} \\
		\and
		David P. Woodruff\\ 
		Department of Computer Science\\
		Carnegie Mellon University\\
		\texttt{dwoodruf@andrew.cmu.edu}
}
\date{}

\begin{document}

\maketitle

\begin{abstract}
We consider the following oblivious sketching problem: given $\epsilon \in (0,1/3)$
and $n \geq d/\epsilon^2$, 
design a distribution $\mathcal{D}$ over $\mathbb{R}^{k \times nd}$ and a function 
$f: \mathbb{R}^k \times \mathbb{R}^{nd} \rightarrow \mathbb{R}$, so that for any $n \times d$ matrix $A$,
$$\Pr_{S \sim \mathcal{D}} [(1-\epsilon) \|A\|_{op} \leq f(S(A),S) \leq (1+\epsilon)\|A\|_{op}] \geq 2/3,$$
where $\|A\|_{op} = \sup_{x:\|x\|_2 = 1} \|Ax\|_2$ is the operator norm of $A$ and $S(A)$ denotes
$S \cdot A$, interpreting $A$ as a vector in $\mathbb{R}^{nd}$. We show a tight lower
bound of $k = \Omega(d^2/\epsilon^2)$ for this problem. Previously, Nelson and Nguyen (ICALP, 2014)
considered the problem of finding a distribution $\mathcal{D}$ over $\mathbb{R}^{k \times n}$ such
that for any $n \times d$ matrix $A$, 
$$\Pr_{S \sim \mathcal{D}}[\forall x, \ (1-\epsilon)\|Ax\|_2 \leq \|SAx\|_2 \leq (1+\epsilon)\|Ax\|_2]
\geq 2/3,$$ which is called an oblivious subspace embedding (OSE). Our result considerably strengthens
theirs, as it (1) applies only to estimating the operator norm, which can be estimated given
any OSE, and (2) applies to distributions
over general linear operators $S$ which treat $A$ as a vector and compute $S(A)$, 
rather than the restricted class of linear operators corresponding to
matrix multiplication. 
Our technique also implies the first tight bounds for approximating the Schatten $p$-norm
for even integers $p$ via general linear sketches, improving the previous lower bound
from $k = \Omega(n^{2-6/p})$ [Regev, 2014] 
to $k = \Omega(n^{2-4/p})$. 
Importantly, for sketching the operator norm up to a factor of $\alpha$, where $\alpha - 1 = \Omega(1)$, 
we obtain a tight $k = \Omega(n^2/\alpha^4)$ bound, matching the upper bound of 
Andoni and Nguyen (SODA, 2013), and improving the previous $k = \Omega(n^2/\alpha^6)$ lower bound.
Finally, we also obtain the first lower bounds for approximating Ky Fan norms. 
\end{abstract}

\maketitle

\section{Introduction}\label{sec:intro}
Understanding the sketching complexity
of estimating matrix norms \cite{BS15,LNW14,Regev14,w14} has been a goal of recent work, generalizing 
a line of work on estimating frequency moments in the sketching
model~\cite{anpw13,lw13,pw12}, and in the somewhat related streaming model of computation~\cite{ams99}. 

In the sketching model, one fixes a distribution
$\mathcal{D}$ over $k\times (nd)$ matrices $S$, and is then given an $n \times d$ matrix $A$ which, without loss of generality, 
satisfies $n \geq d$. 
One then samples $S$ from $\mathcal{D}$, and computes $S(A)$, which denotes the operation of treating $A$ as a
column vector in $\mathbb{R}^{nd}$ and left-multiplying that vector by the matrix $S$. 
Any linear transformation
applied to $A$ can be expressed in this form, and therefore we sometimes refer to such a distribution $\mathcal{D}$
as a {\it general linear sketch}. There is also the related notion of a {\it bilinear sketch}, in which one fixes a 
distribution $\mathcal{D}$ over $k \times n$ matrices $S$, and is then given an $n \times d$ matrix $A$. One
samples $S$ from $\mathcal{D}$ and computes $S \cdot A$. Bilinear sketches are special cases of general linear sketches since they form a subclass of all possible linear transformations of $A$,
and general linear sketches can be much more powerful than bilinear sketches. 
For example, to compute the trace exactly of an $n \times n$
matrix $A$, setting $k = 1$ suffices for a general linear sketch, while we do not know how to compute the trace
with a small value $k$ for bilinear sketches, and several lower bounds on $k$ are known even to approximate the trace 
\cite{wwz14}.

The goal
in the sketching model 
is to minimize the {\it sketching dimension} $k$ so that $S(A)$ can be used to approximate a property of $A$ with constant probability. Associated with distribution $\mathcal{D}$ is an estimation procedure, which we model as a function $f$,
for which $f(S(A), S)$ outputs a correct answer to the problem at hand with constant probability. For numerical properties, such as estimating a norm of $A$, this probability can be amplified to $1-\delta$, by creating a distribution $\mathcal{D}'$ corresponding to taking $O(\log (1/\delta))$ independent copies $S^1, \ldots, S^{\log(1/\delta)}$ from $\mathcal{D}$, and outputting the median of 
$$f(S^1(A), S^1), f(S^2(A), S^2), \ldots, f(S^{\log(1/\delta)}(A), S^{\log(1/\delta)}).$$
Notice that the mapping $S$ is linear and
oblivious, both of which are important for a number of applications such as merging sketches in distributed
computation, or for approximately recovering a signal in compressed sensing. 
Minimizing $k$ is crucial for these applications, as it corresponds to the communication 
or number of observations of the underlying algorithm. 

A quantity of interest is the operator norm. Given a matrix $A$, the operator norm $\|A\|_{op}$ 
is defined to be $\|A\|_{op} =\sup_{x:\|x\|_2=1} \|Ax\|_2$. 
The operator norm arises in several applications; for example 
one sometimes approximates a matrix $A$ by another matrix $\hat{A}$ for which
$\|A -\hat{A}\|_2$ is small. Often $\hat{A}$ has low rank, in which case this is the low rank approximation
problem with spectral error, see, e.g., recent work on this \cite{mm15}. 
If one had an estimator for the operator norm of 
$A - \hat{A}$, one could
use it to verify if $\hat{A}$ is a good approximation to $A$. Given the linearity in the sketching model,
if $S$ is sampled from a distribution $D$, one can compute $S(A) - S(\hat{A}) = S(A - \hat{A})$, from
which one then has an estimation procedure to estimate 
$\|A-\hat{A}\|_2$ as $f(S(A-\hat{A}), S)$. 
In the sketching model, it was first 
shown that approximating the operator norm up to a constant factor requires $k = \Omega(d^{3/2})$ \cite{LNW14}, 
which was later improved by Regev to the tight $k = \Omega(d^2)$ \cite[Section 6.2]{w14}. Note that these
lower bounds rule out {\it any} possible function $f$ as the estimation procedure. 
It is also implicit in \cite[Section 6.2]{w14} that approximating the operator norm up to a factor 
$\alpha$, where $\alpha - 1 = \Omega(1)$, requires $k = \Omega(d^2/\alpha^6)$. 
Andoni and Nguyen showed an upper bound of $k = O(d^2/\alpha^4)$ \cite{AN13}, that is, they constructed a distribution $\mathcal{D}$ and corresponding estimation procedure $f$ for which it suffices to set $k = O(d^2/\alpha^4)$. This follows by Theorem 1.2 of \cite{AN13}.

A wide class of matrix norms is the Schatten $p$-norms, which are the analogues of 
$\ell_p$-norms of vectors 
and contain the operator norm as a special case. The Schatten $p$-norm of matrix $A$ is denoted 
by $\|A\|_p$ and defined to be 
$\|A\|_p = (\sum_{i=1}^n (\sigma_i(A))^p)^{1/p}$, where $\sigma_1,\dots,\sigma_n$ 
are the singular values of $A$. When $p < 1$, $\|A\|_p$ is not a norm but still a well-defined quantity. For $p=0$, 
viewing $\|A\|_0$ as the limit $\lim_{p\to 0^+} \|A\|_p$ recovers exactly the {\it rank} of $A$, which has 
been studied in the data stream~\cite{BS15,cw09} and property testing models~\cite{ks03,lww14}. 
When $p = 1$, it is the nuclear or trace norm\footnote{The trace norm is not to be confused with
the trace. These two quantities only coincide if $A$ is positive semidefinite.}, which has applications 
in differential privacy~\cite{hlm10,lm12} 
and non-convex optimization~\cite{cr12,dtv11}. When $p = 2$, it is the Frobenius norm, and when $p\to\infty$, 
it holds that $\|A\|_p$ tends to $\|A\|_{op}$. 
Such norms are useful in geometry and linear algebra, see, e.g.,~\cite{w14}. 
A $k = \Omega(\sqrt{d})$ lower bound for every $p \geq 0$ was shown in~\cite{lnw14b}. For
$p > 2$ a lower bound on the sketching dimension of $k = \Omega(d^{2/3-3/p})$, and an upper bound of 
$k = O(d^{2-4/p})$ were shown in~\cite{lnw14b}. 
The upper bound is only known to hold when $p$ is an even integer. The lower bound was
improved by Regev to 
$k = \Omega(d^{2-6/p})$ for $p > 6$ \cite[Section 6.2]{w14}\footnote{The section discusses only the case of $p=\infty$, i.e., the operator norm, but the same method can be used for general $p$ and gives the bound claimed here.}.

Other related work includes that on {\it oblivious subspace embeddings} (OSEs), which fall into the category
of bilinear sketches.  
Here one seeks a distribution $\mathcal{D}$ over $\mathbb{R}^{k \times n}$ such
that for any $n \times d$ matrix $A$, 
$$\Pr_{S \sim \mathcal{D}}[\forall x, \ (1-\epsilon)\|Ax\|_2 \leq \|SAx\|_2 \leq (1+\epsilon)\|Ax\|_2]
\geq 2/3.$$
This notion has proved important in numerical linear algebra, and has led to the fastest
known algorithms for low rank approximation and regression \cite{cw13,mm13,nn13}. Since an OSE
has the property that $\|SAx\|_2 = (1 \pm \epsilon) \|Ax\|_2$ for all $x$, it holds in particular
that $\|SA\|_{op} = (1 \pm \epsilon) \|A\|_{op}$, where the notation $a = (1 \pm \epsilon) b$ means
$(1-\epsilon) b \leq a \leq (1+\epsilon)b$. When $n \geq d/\epsilon^2$, Nelson and Nguyen show the
tight bound that any OSE requires $k = \Omega(d/\epsilon^2)$ \cite{NN14}. 

Finally, we mention recent related work in the data stream model on approximation of matrix norms
\cite{BS15,lw16}. 
Here one sees
elements of $A$ one at a time and the goal is to output an approximation to $\|A\|_p$. 
It is important to note that the data stream model and sketching models are incomparable. The main
reason for this is that unlike in the data stream model, 
the bit complexity is not accounted for in the sketching model, 
and both $S$ and $A$ are assumed to have entries which are real numbers. The latter is the common
model adopted in compressed sensing. In the data stream model, if one wants to output a vector
$v \in \{0, 1, \ldots, M-1, M\}^n$, one needs $n \log M$ bits of space. On the other hand, if
$u$ is the vector $(1, (M+1), (M+1)^2, (M+1)^3, \ldots, (M+1)^n)$, then from $\langle u, v \rangle$,
one can output $v$, so the sketching dimension $k$ is only equal to $1$. The sketching complexity
thus gives a meaningful measure of complexity in the real RAM model. Conversely, lower bounds in
the sketching model do not translate into lower bounds in the data stream model.
This statement holds even given the 
work of \cite{LNW14} which characterizes turnstile streaming algorithms as linear sketches. The
problem is that lower bounds in the sketching model 
involve continuous distributions and after discretizing the distributions it is no longer clear if the 
lower bounds hold.

\subsection{Our Contributions}
In this paper we strengthen known sketching lower bounds for the operator norm, 
Schatten $p$-norms, and subspace embeddings. Our lower bounds are optimal for any approximation to 
the operator norm, for subspace embeddings, and for Schatten $p$-norms for even integers $p$. 
We first describe our results for the operator norm, as the results for Schatten $p$-norms and 
subspace embeddings follow from them. 

We consider the following problem: given $\epsilon \in (0,1/3)$
and $n \geq d/\epsilon^2$, 
design a distribution $\mathcal{D}$ over $\mathbb{R}^{k \times nd}$ and a function $f: \mathbb{R}^k \times \mathbb{R}^{k \times nd} \rightarrow \mathbb{R}$, so that for any $n \times d$ matrix $A$,
$$\Pr_{S \sim \mathcal{D}} [(1-\epsilon) \|A\|_{op} \leq f(S(A), S) \leq (1+\epsilon)\|A\|_{op}] \geq 2/3,$$
For this problem, we show a tight $k = \Omega(d^2/\epsilon^2)$ lower bound. 
Our result considerably strengthens
the result of Nelson and Nguyen \cite{NN14} as it 
(1) applies only to estimating the operator norm, which can be estimated given
any OSE, and (2) applies to general linear sketches rather than only to bilinear sketches. 
Regarding (1), this 
shows that designing a general linear sketch for approximating the operator norm of a matrix is 
{\it as hard as designing an oblivious subspace embedding}.
Regarding (2), we lower bound a much larger class of data structures than OSEs that
one could
use to approximate $\|Ax\|_2$ for all vectors $x$. 

We then generalize the argument above to handle approximation factors $\alpha$, with $\alpha - 1 = \Omega(1)$, 
for approximating the operator norm. In this case we consider $n = d$, which is without loss of generality
since by first applying an OSE $S$ to $A$ with $k = O(d)$, replacing $A$ with $S \cdot A$, all singular
values of $A$ are preserved up to a constant factor (we can also pad $SA$ with zero columns to make 
$SA$ be a square matrix) - see Appendix C of \cite{lnw14b}. We can then apply our general linear
sketch to $SA$ (the composition of linear sketches is a general linear sketch). We show a lower
bound of $k = \Omega(n^2/\alpha^4)$, improving the previous $k = \Omega(n^2/\alpha^6)$ bound,
and maching the $k = O(n^2/\alpha^4)$ upper bound. This answers Open Question 2 in \cite{lnw14b}.  

The proof shows the problem is already hard
to distinguish between the two cases: (1) $A$ has one singular value of value $\Theta(\alpha)$ 
and remaining singular values of
value $\Theta(1)$, versus (2) all singular values of $A$ are of value $\Theta(1)$. 
By setting $\alpha = n^{1/p}$, we are able to obtain a constant factor gap
in the Schatten-$p$ norm in the two cases, and therefore additionally obtain an $\Omega(n^{2-4/p})$
lower bound for Schatten $p$-norms for constant factor approximation. This improves the previous
$\Omega(n^{2-6/p})$ lower bound, and matches the known upper bound for even integers $p$. Our proof
also establishes a lower bound of $k = \Omega(n^2/s^2)$ 
for estimating the Ky-Fan $s$-norm of an $n \times n$
matrix $A$ up to a constant factor, whenever $s \leq .0789\sqrt{n}$. 

Our main technical novelty is avoiding a deep 
theorem of Lata\l{}a~\cite{latala} concerning tail bounds for Gaussian chaoses
used in the prior lower bounds for sketching the operator norm and Schatten $p$-norms. 
Instead we prove a simple lemma (Lemma~\ref{lem:Eexp(x^TAy)}) allowing us
to bound $\mathbb{E}_{x,y}[e^{x^T Ay}]$ for Gaussian vectors $x$ and $y$ and a 
matrix $A$, 
in terms of the Frobenius norm 
of $A$. Surprisingly, this lemma
suffices for directly upper-bounding the $\chi^2$-distance between the distributions 
considered in previous works, and without losing any additional factors. 
Our technical arguments
are thus arguably more elementary and simpler than those given in previous work. 

\section{Preliminaries}\label{sec:prelim}

\subsection*{Notation.} Let $\R^{n\times d}$ be the set of $n\times d$ real matrices and  $N(\mu,\Sigma)$ denote the (multi-variate) normal distribution of mean $\mu$ and covariance matrix $\Sigma$. 
We write $X\sim \mathcal{D}$ for a random variable $X$ subject to a probability distribution $\mathcal{D}$. 
Denote by $\G(n,n)$ the ensemble of random matrices with entries i.i.d.\ $N(0,1)$. 
	
\subsection*{Singular values and matrix norms.} Consider a matrix $A\in \R^{n \times n}$. Then $A^TA$ is a positive semi-definite matrix. The eigenvalues of $\sqrt{A^TA}$ are called the singular values of $A$, denoted by $\sigma_1(A)\geq \sigma_2(A)\geq \cdots \geq\sigma_n(A)$ in decreasing order. Let $r=\rk(A)$. It is clear that $\sigma_{r+1}(A) = \cdots = \sigma_n(A) = 0$. 
Define $
\|A\|_p = (\sum_{i=1}^r (\sigma_i(A))^p)^{1/p}
$ ($p>0$).
For $p\geq 1$, it is a norm over $\R^{n\times d}$, called the $p$-th \textit{Schatten norm}, over $\R^{n\times n}$ for $p\geq 1$. When $p=1$, it is also called the trace norm or nuclear norm. When $p=2$, it is exactly the Frobenius norm $\|A\|_F$. 
Let $\|A\|_{op}$ denote the operator norm of $A$ when treating $A$ as a linear operator from $\ell_2^n$ to $\ell_2^n$. It holds that $\lim_{p\to\infty} \|A\|_p = \sigma_1(A) = \|A\|_{op}$.

The Ky-Fan $s$-norm of $A$, denoted by $\|A\|_{F_s}$, is defined as the sum of the largest $s$ singular values: $\|A\|_{F_s} = \sum_{i=1}^s \sigma_i(A)$. Note that $\|A\|_{F_1} = \|A\|_{op}$ and $\|A\|_{F_s} = \|A\|_1$ for $s\geq r$.

\subsection*{Distance between probability measures.} Suppose $\mu$ and $\nu$ are two probability measures over some Borel algebra $\mathcal{B}$ on $\R^n$ such that $\mu$ is absolutely continuous with respect to $\nu$. For a convex function $\phi:\R\to \R$ such that $\phi(1)=0$, we define the $\phi$-divergence
\[
D_\phi(\mu || \nu) = \int \phi\left(\frac{d\mu}{d\nu}\right)d\nu.
\]
In general $D_\phi(\mu||\nu)$ is not a distance because it is not symmetric. 

The \textit{total variation distance} between $\mu$ and $\nu$, denoted by $d_{TV}(\mu,\nu)$, is defined as $D_\phi(\mu||\nu)$ for $\phi(x) = |x-1|$. It can be verified that this is indeed a distance.

The \textit{$\chi^2$-divergence} between $\mu$ and $\nu$, denoted by $\chi^2(\mu||\nu)$, is defined as $D_\phi(\mu||\nu)$ for $\phi(x) = (x-1)^2$ or $\phi(x) = x^2-1$. It can be verified that these two choices of $\phi$ give exactly the same value of $D_\phi(\mu||\nu)$.

\begin{proposition}[{\cite[p90]{Tsybakov}}] \label{prop:TV_chi^2}
$d_{TV}(\mu,\nu) \leq \sqrt{\chi^2(\mu||\nu)}$.
\end{proposition}

\begin{proposition}[{\cite[p97]{IS}}] \label{prop:chi^2}
$\chi^2(N(0,I_n)\ast \mu||N(0,I_n)) \leq \E e^{\langle x,x'\rangle}-1$, where $x,x'\sim \mu$ are independent.
\end{proposition}

\section{Sketching Lower Bound for $p > 2$}\label{sec:sketching_p>2}
We follow the notations in \cite{lnw14b} throughout this section, though
the presentation here is self-contained. To start, we present the following lemma.

\begin{lemma}[\footnote{A similar result holds for subgaussian vectors $x$ and $y$ with the right-hand side replaced with $\exp(c\|A\|_F^2)$ for some absolute constant $c>0$, whose proof requires heavier machinery. We only need the elementary variant here by our choice of hard instance.}]\label{lem:Eexp(x^TAy)}
Suppose that $x\sim N(0,I_m)$ and $y\sim N(0,I_n)$ are independent and $A\in \R^{m\times n}$ satisfies $\|A\|_F < 1$. It holds that
\[
\E_{x,y} e^{x^T Ay} \leq \frac{1}{\sqrt{1-\|A\|_F^2}}.
\]
\end{lemma}
\begin{proof}
First, it is easy to verify that
\begin{align*}
\E_{x,y\sim N(0,1)} e^{axy} &= \frac{1}{2\pi}\iint_{\R\times \R} e^{axy-\frac{x^2+y^2}{2}} dxdy\\
&= \frac{1}{2\pi}\int_{\R}\int_{\R} e^{-\frac{1}{2} (x-ay)^2}e^{-\frac{1}{2}(1-a^2)y^2} dx dy\\
&= \frac{1}{\sqrt{2\pi}}\int_{\R} e^{-\frac{1}{2}(1-a^2)y^2} dy\\
& = \frac{1}{\sqrt{1-a^2}},\quad a\in [0,1).
\end{align*}
Without loss of generality, assume that $m\geq n$. Consider the singular value decomposition $A = U\Sigma V^T$ where $U$ and $V$ are orthogonal matrices of dimension $m$ and $n$ respectively and $\Sigma = \diag\{\sigma_1,\dots,\sigma_n\}$ with $\sigma_1,\dots,\sigma_n$ being the non-zero singular values of $A$. We know that $\sigma_i \in [0,1)$ for all $i$ by the assumption that $\|A\|_F < 1$. By rotational invariance of the Gaussian distribution, we may assume that $m = n$ and thus
\begin{align*}
\E_{x,y\sim N(0,I_n)} e^{x^TAy} &= \E_{x,y\sim N(0,I_n)} e^{x^T\Sigma y}\\
&= \frac{1}{(2\pi)^{n}}\iint_{\R^n\times \R^n} \exp\left\{\sum_{i=1}^n \left(\sigma_i x_i y_i - \frac{x_i^2+y_i^2}{2}\right)\right\} dxdy\\
&= \prod_{i=1}^n \frac{1}{\sqrt{1-\sigma_i^2}}\\ 
&\leq \frac{1}{\sqrt{1-\sum_{i=1}^n \sigma_i^2}}\\
&= \frac{1}{\sqrt{1- \|A\|_F^2}}.\qedhere
\end{align*}
\end{proof}

Next we consider the problem of distinguishing two distributions $\D_1 = \G(m,n)$ and $\D_2$ as defined below. Let $u_1,\dots,u_r$ be i.i.d.~$N(0,I_m)$ vectors and $v_1,\dots,v_r$ i.i.d.~$N(0,I_n)$ vectors and further suppose that $\{u_i\}$ and $\{v_i\}$ are independent. Let $s\in \R^r$ and define the distribution $\D_2$ as $\G(m,n) + \sum_{i=1}^r s_i u^i (v^i)^T$. We take $k$ linear measurements and denote the corresponding rows (measurements) of the sketching matrix by $L^1,\dots,L^k$. Without loss of generality we may assume that $\tr((L^i)^T L^i))=1$ and $\tr((L^i)^T L^j))=0$ for $i\neq j$, since this corresponds to the rows of the sketching matrix being orthonormal, which we can assume since we can always change the basis of the row space of the sketching matrix in a post-processing step. Let $\LL_1$ and $\LL_2$ be the corresponding distribution of the linear sketch of dimension $k$ on $\D_1$ and $\D_2$, respectively. The main result is the following theorem.

\begin{theorem}\label{thm:sketch_lower_bound} There exists an absolute constant $c>0$ such that $d_{TV}(\LL_1,\LL_2)\leq 1/10$ whenever $k\leq c/\|s\|_2^4$.
\end{theorem}
\begin{proof}
It is not difficult to verify that $\LL_1 = N(0,I_k)$ and $\LL_2 = N(0,I_k) + \mu$, where $\mu$ is the distribution of
\[
\begin{pmatrix}
\sum_{i=1}^r s_i (u^i)^T L^1 v^i\\
\sum_{i=1}^r s_i (u^i)^T L^2 v^i\\
\vdots\\
\sum_{i=1}^r s_i (u^i)^T L^k v^i\\
\end{pmatrix}.
\]
Consider a random variable (we shall see in a moment where it comes from)
\[
\xi = \sum_{i=1}^k \sum_{j, l=1}^r \sum_{a,c=1}^m \sum_{b,d=1}^n s_j s_l (L^i)_{ab}(L^i)_{cd} (u^j)_a (v^j)_b (u^l)_c (v^l)_d.
\]
Take expectation on both sides and notice that the non-vanishing terms on the right-hand side must have $j=l$, $a=c$ and $b=d$, 
\[
\E\xi = \sum_{i=1}^k \sum_{j=1}^r \sum_{a=1}^m \sum_{b=1}^n s_j^2 (L^i)_{ab}^2 \E(u^j)_a^2 \E(v^j)_a^2 = k\|s\|_2^2.
\]
Define an event $\mathcal{E} = \{ \|s\|^2\xi < 1/2\}$ and it follows from our assumption and Markov's inequality that $\Pr(\mathcal{E}) \geq 1-2c$. Restrict $\mu$ to this event and denote the induced distribution by $\tilde\mu$. Let $\tilde \LL_2 = N(0,I_n) + \tilde\mu$.

Then the total variation distance between $\LL_1$ and $\LL_2$ can be upper bounded as
\begin{align*}
d_{TV}(\LL_1,\LL_2) &\leq d_{TV}(\LL_1,\tilde\LL_2) + d_{TV}(\LL_2,\tilde\LL_2)\\
&\leq \sqrt{\E_{z_1,z_2\sim \tilde\mu}e^{\langle z_1,z_2\rangle} - 1} + d_{TV}(\mu,\tilde\mu)\\
&\leq \sqrt{\frac{1}{\Pr(\mathcal{E})}(\E_{z_1\sim \tilde\mu, z_2\sim\mu}e^{\langle z_1,z_2\rangle} - 1)} + \frac{1}{\Pr(\mathcal{E})}-1
\end{align*}
and we shall bound $\E e^{\langle z_1,z_2\rangle}$ in the rest of the proof.
\begin{align*}
\E_{z_1\sim\tilde\mu,z_2\sim\mu} e^{\langle z_1,z_2\rangle} &= \E \exp\left\{ \sum_{i=1}^k \sum_{j,a,b} \sum_{j',a',b'} s_j (L^i)_{ab} (u^j)_a (v^j)_b \cdot s_{j'} (L^i)_{a'b'} (x^{j'})_{a'} (y^{j'})_{b'} \right\}\\
&= \E_{u^1,\dots,u^r,v^1\dots,v^r|\tilde\mu} \prod_{j'=1}^r \E_{\substack{x_{j'}\sim N(0,I_m)\\y_{j'}\sim N(0,I_n)}} \exp\left\{ \sum_{a',b'} Q^{j'}_{a',b'} (x^{j'})_{a'} (y^{j'})_{b'} \right\},
\end{align*}
where 
\[
Q^{j'}_{a',b'}  = s_{j'}\sum_{i=1}^k \sum_{j,a,b} (L^i)_{ab}(L^i)_{a'b'}\cdot s_j  (u^j)_a (v^j)_b.
\]
In order to apply the preceding lemma, we need to verify that $\|Q^{j'}\|_F^2 < 1$. Indeed,
\begin{align*}
\|Q^{j'}\|_F^2 &= \sum_{a',b'} (Q^{j'})_{a',b'}^2\\
&= s_{j'}^2 \sum_{a',b'} \sum_{i,i'} \sum_{j,a,b}\sum_{\ell,c,d} s_j(L^i)_{ab}(L^i)_{a'b'}(u^j)_a (v^j)_b \cdot s_\ell(L^{i'})_{cd}(L^{i'})_{a'b'}(u^\ell)_c (v^\ell)_d\\
&= s_{j'}^2 \sum_{a',b'} \sum_{i} (L^i)^2_{a'b'} \sum_{j,a,b}\sum_{\ell,c,d} s_j(L^i)_{ab}(u^j)_a (v^j)_b \cdot s_\ell(L^{i})_{cd}(u^\ell)_c (v^\ell)_d \quad (i\text{ must equal to }i')\\
&= s_{j'}^2 \sum_{i} \sum_{j,a,b}\sum_{\ell,c,d} s_j(L^i)_{ab}(u^j)_a (v^j)_b \cdot s_\ell(L^{i})_{cd}(u^\ell)_c (v^\ell)_d\\
&= s_{j'}^2 \xi < 1
\end{align*}
since we have conditioned on $\mathcal{E}$. Now it follows from the preceding lemma that
\begin{align*}
\E_{u^1,\dots,u^r,v^1\dots,v^r} \prod_{i=1}^r \E_{x_{j'},y_{j'}} \exp\left\{ \sum_{a',b'} Q^{j'}_{a',b'} (x^{j'})_{a'} (y^{j'})_{b'} \right\} 
&\leq \E_{u^1,\dots,u^r,v^1\dots,v^r} \prod_{j'=1}^r \frac{1}{\sqrt{1-s_{j'}^2\xi}}\\
&\leq \E_{u^1,\dots,u^r,v^1\dots,v^r} \frac{1}{\sqrt{1-\|s\|^2\xi}} \\
&\leq 1+\|s\|^2\E\xi \\
&\leq 1 + k\|s\|^4,
\end{align*}
where, in the third inequality, we used the fact that $1/\sqrt{1-x}\leq 1+x$ for $x\in [0,1/2]$.
Therefore,
\[
d_{TV}(\LL_1,\LL_2) \leq \sqrt{\frac{k\|s\|^4}{1-2c}} + \frac{2c}{1-2c} 
\leq \sqrt{\frac{c}{1-2c}} + \frac{2c}{1-2c} 
\leq \frac{1}{10}
\]
when $c > 0$ is small enough.
\end{proof}

We will apply the preceding theorem to obtain our lower bounds for the 
applications. To do so, notice that by Yao's minimax principle, we can
fix the rows of our sketching matrix, and show that the resulting distributions
$\mathcal{L}_1$ and $\mathcal{L}_2$ above have small total variation distance.
By standard properties of the variation distance, 
this implies that no estimation procedure $f$ can be used to distinguish
the two distributions with sufficiently large probability, thereby establishing
our lower bound. 

\begin{corollary}[$\alpha$-approximation to operator norm] Let $c>0$ be an arbitrarily small constant. For $\alpha \geq 1+c$, any sketching algorithm that estimates $\|X\|_{op}$ for $X\in \R^{n\times n}$ within a factor of $\alpha$ with error probability $\leq 1/6$ requires sketching dimension $\Omega(n^2/\alpha^4)$.
\end{corollary}
\begin{proof}
Let $m=n$ and take $r=1$ and $s_1 = C\alpha/\sqrt{n}$ for some constant $C$ large enough in $\D_2$ and apply the preceding theorem.
\end{proof}

\begin{corollary}[Schatten norms] There exists an absolute constant $c>0$ such that any sketching algorithm that estimates $\|X\|_p^p$ ($p>2$) for $X\in \R^{n\times n}$ within a factor of $1+c$ with error probability $\leq 1/6$ requires sketching dimension $\Omega(n^{2(1-2/p)})$.
\end{corollary}
\begin{proof}
Let $m=n$ and take $r=1$ and $s_1 = 5/n^{1/2-1/p}$ in $\D_2$. Note that $\|X\|_p^p$ differs by a constant factor with high probability when $X\sim \D_1$ and $X\sim \D_2$ (the same hard distribution as in \cite{lnw14b}), apply the preceding theorem.
\end{proof}

\begin{corollary} Let $\epsilon\in (0,1/3)$. For any matrix $X\in \R^{(d/\epsilon^2)\times d}$, any sketching algorithm that estimates $\|X\|_{op}$ within a factor of $1+\epsilon$ with error probability $\leq 1/6$ requires sketching dimension $\Omega(d^2/\epsilon^2)$.
\end{corollary}
\begin{proof}
Let $m = d/\epsilon^2$ and $n = d$. Take $r=1$ and $s_1 = \alpha\sqrt{\epsilon/d}$ for some constant $\alpha > 0$ large enough and apply Theorem~\ref{thm:sketch_lower_bound}. Next we shall justify this choice of parameters, that is,
\[
G \qquad\text{ and }\qquad G + \alpha\sqrt{\frac{\epsilon}{d}} uv^T
\]
differ in operator norm by a factor of $1+\epsilon$ for some constant $\alpha$. This is de facto proved in the proof of Theorem 7.3 of~\cite{LWW21}; nevertheless, we include a full proof below for completeness. First, it follows from the standard result~\cite{V12} that
\[
\|G\|_{op} \leq \frac{\sqrt{d}}{\epsilon} + 1.1\sqrt{d} = (1+1.1\epsilon)\frac{\sqrt{d}}{\epsilon}
\]
with probability at least $1-e^{-\Omega(d)}$. Next we shall show that
\[
\left\|G + \alpha\sqrt{\frac{\epsilon}{d}} uv^T\right\|_{op}\geq (1+2\epsilon)\frac{\sqrt{d}}{\epsilon}
\]
with high probability. Observe that (denoting the unit sphere in $\R^d$ by $\bS^{d-1}$)
\begin{align*}
\left\|G + \alpha\sqrt{\frac{\epsilon}{d}} uv^T\right\|_{op} = \sup_{x\in \bS^{d-1}} \left\|\left(G+ \alpha\sqrt{\frac{\epsilon}{d}} uv^T\right)x\right\|_2 
&\geq \left\|\left(G + \alpha\sqrt{\frac{\epsilon}{d}} uv^T \right)\frac{v}{\|v\|_2}\right\|_2 \\
&= \left\|G\frac{v}{\|v\|_2} + \alpha\sqrt{\frac{\epsilon}{d}} u \|v\|_2 \right\|_2.
\end{align*}
Since $v\sim N(0,I_d)$, the direction $v/\|v\|_2 \sim \Unif(\bS^{d-1})$ and the magnitude $\|v\|_2$ are independent, and by rotational invariance of the Gaussian distribution, $Gx\sim N(0,I_d)$ for any $x\in \bS^{d-1}$. Hence
\[
\left\|G\frac{v}{\|v\|_2} + \alpha\sqrt{\frac{\epsilon}{d}} u \|v\|_2 \right\|_2 \eqdist \left\|u_1 + \alpha\sqrt{\frac{\epsilon}{d}} t u_2\right\|_2 \eqdist \sqrt{1+\frac{\alpha^2 \eps^2 t^2}{d}}\|u\|_2,
\]
where $t$ follows the distribution of $\|v\|_2$, $u_1,u_2,u\sim N(0,I_m)$, $t, u_1, u_2$ are independent, $t$ and $u$ are independent. By the standard results of standard gaussian vectors, with probability at least $1-e^{-\Omega(d)}$, it holds that $t\geq \sqrt{d/2}$ and $\|u\|_2\geq (1-\eps)\sqrt{d}/\eps$. Therefore, with probability at least $1-e^{-\Omega(d)}$, we have
\[
\left\|G + \alpha\sqrt{\frac{\epsilon}{d}} uv^T\right\|_{op}\geq \sqrt{1 + \frac{\alpha^2 \eps}{2}}(1-\eps)\frac{\sqrt{d}}{\eps}\geq (1+2\eps)\frac{\sqrt d}{\eps}
\]
for all $\eps\in(0,1/3)$, provided that $\alpha\geq 3\sqrt{7/2}$.
\end{proof}

\begin{corollary}[Ky-fan norm]\label{cor:ky-fan}
There exists an absolute constant $c>0$ such that any sketching algorithm that estimates $\|X\|_{F_s}$ for $X\in \R^{n\times n}$ and $s\leq 0.0789\sqrt{n}$ within a factor of $1+c$ with error probability $\leq 1/6$ requires sketching dimension $\Omega(n^2/s^2)$.
\end{corollary}
\begin{proof}
Take $r=s$ and $s_1 = s_2 = \cdots = s_r = 5/\sqrt{n}$ in $\D_2$ and apply Theorem~\ref{thm:sketch_lower_bound}, for which we shall show the KyFan $s$-norms are different with high probability in the two cases. 

When $X\sim \mathcal{D}_1$, we know that $\sigma_1(X)\leq 2.1\sqrt{n}$ with high probability and thus $\|X\|_{F_s} \leq 2.1s\sqrt{n}$ with high probability. 

When $X\sim \mathcal{D}_2$, we can write $X = G + \frac{5}{\sqrt n}P$, where $P = u_1v_1^T + \cdots + u_s v_s^T$. We claim that with high probability $\|P\|_1 \geq 0.9sn$ and thus $\|X\|_{F_s}\geq \frac{5}{\sqrt n}\|P\|_{F_s} - \|G\|_{F_s} \geq 4.5s\sqrt{n} - 2.1s\sqrt{n}\geq 2.4s\sqrt{n}$, evincing a multiplicative gap of $\|X\|_{F_s}$ between the two cases. 

Now we prove the claim. With high probability, it holds that $0.99\sqrt{n} \leq \|u_i\|\leq 1.01\sqrt{n}$ for all $i$ and $|\sum_{i\neq j}\langle u_i,u_j\rangle|\leq 1.01s\sqrt{n}$.
%
%
We shall condition on these events below.

By the min-max theorem for singular values,
\[
\sigma_\ell^2(P) = \max_{H:\dim H = \ell} \min_{\substack{x\in H\\ \|x\|_2=1}} x^TP^TPx,
\] 
where
\begin{align*}
x^TP^TPx &= \sum_{i,j} x^T v_i u_i^T u_j v_j^T x\\
	&= \sum_i x^T v_i u_i^T u_i v_i^T x + \sum_{i\neq j} x^T v_i (u_i^T u_j) v_j^T x\\
	&\geq 0.99^2 n\sum_i x^T v_i^T v_i x - 1.01\sqrt{n} \cdot 1.01 k\sqrt{n}\cdot 1.01\sqrt{n}\\
	&= 0.99^2 n\sum_i x^T v_i^T v_i x - 1.01^3 k  n^{\frac{3}{2}}
\end{align*}
and thus, 
\begin{align*}
\sigma_\ell^2(P) &\geq 0.99^2 n \max_{H:\dim H = \ell} \min_{\substack{x\in H\\ \|x\|_2=1}} \sum_i x^T v_i^T v_i x - 1.01^3 k  n^{\frac{3}{2}}\\
&= 0.99^2 n\sigma_\ell^2(V) - 1.01^3 k  n^{\frac{3}{2}},
\end{align*}
where $V$ is a $k\times n$ matrix with rows $v_1^T,\dots, v_k^T$. Therefore
\[
\|P\|_1 \geq 0.99\sqrt{n}\|V\|_1 - 1.01^{\frac32} s^{\frac32} n^{\frac34}.
\]
Since $V$ is a Gaussian random matrix, the classical results imply that $\|V\|_1\geq 0.99s\sqrt{n}$ with high probability \cite{tao}. The claim follows from our assumption on $s$.
\end{proof}

\section{Conclusion}
We have presented a simple, surprisingly 
powerful new analysis which gives optimal bounds 
on the sketching dimension for a number of previously studied sketching problems, 
including approximating 
the operator norm, Schatten norms, and subspace embeddings. We have also presented
the first lower bounds for estimating Ky Fan norms. 
It would be interesting to see if there are other applications of this method to
the theory of linear sketches.

\bibliographystyle{plain}
\bibliography{literature}

\end{document}